\documentclass[10pt,conference]{IEEEtran}

\usepackage{amssymb,stmaryrd,amsmath,amsfonts,rotating}
\usepackage[noadjust]{cite}
\usepackage{color}
\usepackage[vflt]{floatflt}
\usepackage{epic}
\RequirePackage{bbm}
\definecolor{TODO}{rgb}{0.6,0.6,0.6} 

\definecolor{TOCHECK}{rgb}{0.8,0.8,0.8} 

\newtheorem{theorem}{Theorem}
\newcommand{\btheo}{\begin{theorem}}
\newcommand{\etheo}{\end{theorem}}
\newcommand{\bproof}{\begin{proof}}
\newcommand{\eproof}{\end{proof}}
\newtheorem{definition}[theorem]{Definition}
\newcommand{\bdefi}{\begin{definition}}
\newcommand{\edefi}{\end{definition}}
\newtheorem{fact}[theorem]{Fact}
\newcommand{\bprop}{\begin{fact}}
\newcommand{\eprop}{\end{fact}}
\newtheorem{corollary}[theorem]{Corollary}
\newcommand{\bcor}{\begin{corollary}}
\newcommand{\ecor}{\end{corollary}}
\newtheorem{example}[theorem]{Example}
\newcommand{\bex}{\begin{example}}
\newcommand{\eex}{\end{example}}
\newtheorem{lemma}[theorem]{Lemma}
\newcommand{\blemma}{\begin{lemma}}
\newcommand{\elemma}{\end{lemma}}
\newtheorem{remark}[theorem]{Remark}
\newcommand{\bremark}{\begin{remark}}
\newcommand{\eremark}{\end{remark}}
\newtheorem{conj}[theorem]{Conjecture}
\newcommand{\bconj}{\begin{conj}}
\newcommand{\econj}{\end{conj}}

\def\0{{\tt 0}} 
\def\1{{\tt 1}} 
\def\?{{\tt *}} 
\newcommand{\code}{{\ensuremath{\tt C}}}

\renewcommand{\mid}{\,|\,}

\newcommand {\wt} {{\mathtt{ wt}}}

\newcommand {\RM} {{\text{RM}}}
\newcommand {\disto}{\mathtt{d}}
\newcommand {\Zd}{\widehat{Z}}

\begin{document}
\title{Performance of Polar Codes for Channel and Source Coding}
\author{
\authorblockN{Nadine Hussami}
\authorblockA{
AUB, Lebanon,
Email: njh03@aub.edu.lb}
\and
\authorblockN{Satish Babu Korada and R\"udiger Urbanke}
\authorblockA{
EPFL, Switzerland,
Email: \{satish.korada,ruediger.urbanke\}@epfl.ch}
}
\maketitle
\begin{abstract}
Polar codes, introduced recently by Ar\i kan, are the first
family of codes known to achieve capacity of symmetric channels using a
low complexity successive cancellation decoder.  Although these codes,
combined with successive cancellation, are optimal in this respect,
their finite-length performance is not record breaking. We discuss several
techniques through which their finite-length performance can be improved.
We also study the performance of these codes in the context of source
coding, both lossless and lossy, in the single-user context as well as
for distributed applications.  
\end{abstract}

\section{Introduction}
Polar codes, recently introduced by Ar\i kan in \cite{Ari08},
are the first provably capacity achieving family of codes for arbitrary
symmetric binary-input discrete memoryless channels (B-DMC) with low
encoding and decoding complexity. The construction of polar codes
is based on the following observation: Let
$G_2 = \bigl[ \begin{smallmatrix} 1 &0 \\ 1& 1 \end{smallmatrix} \bigr]$.
Apply the transform $G_2^{\otimes n}$ (where ``$\phantom{}^{\otimes n}$''
denotes the $n^{th}$ Kronecker power) to a block of $N = 2^n$ bits and
transmit the output through independent copies of a symmetric B-DMC, call it
$W$. As $n$ grows large, the channels seen by individual bits (suitably
defined in \cite{Ari08}) start \emph{polarizing}: they approach either a
noiseless channel or a pure-noise channel, where the fraction of channels
becoming noiseless is close to the capacity $I(W)$.
In the following, let $\bar{u}$ denote the vector $(u_0,\dots,u_{N-1})$
and $\bar{x}$ denote the vector $(x_0,\dots,x_{N-1})$.  Let
$\pi:\{0,\dots,N-1\}\to \{0,\dots,N-1\}$ be the permutation such that if
the $n$-bit binary representation of $i$ is $b_{n-1}\dots b_0$, then
$\pi(i) = b_{0}\dots b_{n-1}$ (we call this a {\em bit-reversal}).
Let $\wt(i)$ denote the number of ones in the binary expansion of $i$.

\subsection{Construction of Polar Codes}
The channel polarization phenomenon suggests to use the noiseless channels
for transmitting information while fixing the symbols transmitted
through the noisy ones to a value known both to sender as well as
receiver. For symmetric channels we can assume without loss of
generality that the fixed positions are set to $0$. Since the fraction of
channels becoming noiseless tends to $I(W)$, this scheme achieves the capacity
of the channel.

In \cite{Ari08} the following alternative interpretation was mentioned;
the above procedure can be seen as transmitting a codeword and decoding
at the receiver with a successive cancellation (SC) decoding strategy. The specific code
which is constructed can be seen as a generalization of the Reed-Muller
(RM) codes.
Let us briefly discuss the construction of RM codes. We follow the
lead of \cite{For01} in which the Kronecker product is used.  RM codes
are specified by the two parameters $n$ and $r$ and the
code is denoted by $\RM(n,r)$. An RM($n,r$) code has block length
$2^n$ and rate $\frac{1}{2^n} \sum_{i=0}^r {n \choose i}$. The code is
defined through its generator matrix as follows. Compute the Kronecker
product $G_2^{\otimes n}$. This gives a $2^n \times 2^n$ matrix. Label the
rows of this matrix as $0,\dots,2^n-1$. One can check that the weight of the 
$i$th row of this matrix is equal to $2^{\wt(i)}$. The
generator matrix of the code $RM(n,r)$ consists of all the rows of
$G_2^{\otimes n}$ which have weight at least $2^{n-r}$. There are exactly
$\sum_{i=0}^r {n \choose i}$ such rows.  An equivalent way
of expressing this is to say that the codewords are of the form $\bar{x}
= \bar{u}G_2^{\otimes n}$, where the components $u_i$ of $\bar{u}$
corresponding to the rows of $G_2^{\otimes n}$ of weight less than $2^{n-r}$
are fixed to $0$ and the remaining components contain the ``information."
Polar codes differ from RM codes only in the choice of generator
vectors $\bar{u}$, i.e., in the choice of which components of  $\bar{u}$
are set to $0$.  Unlike RM codes, these codes are defined for any
dimension $1\leq k \leq 2^n$.  The choice of the generator vectors, as
explained in \cite{Ari08}, is rather complicated; we therefore do not
discuss it here.  Following Ar\i kan, call those components $u_i$ of
$\bar{u}$ which are set to $0$ ``frozen," and the remaining ones 
``information" bits.  Let the set of frozen bits be denoted by $F$
and the set of information bits be denoted by $I$. A polar
code is then defined as the set of codewords of the form $\bar{x} =
\bar{u}G_2^{\otimes n}$, where the bits $i\in F$ are fixed to $0$.

\blemma[RM Code as Polar Code]
A $\RM(n,r)$ is a polar code of length $2^n$ with $F=\{u_i: \wt(i) < r\}$. 
\elemma

\subsection{Performance under Successive Cancellation Decoding}
In \cite{Ari08} Ar\i kan considers a low complexity SC decoder.  We briefly
describe the decoding procedure here.  The bits
are decoded in the order $\pi(0),\dots,\pi({N-1})$. Let the estimates of
the bits be denoted by $\hat{u}_0,\dots,\hat{u}_{N-1}$. If a bit $u_i$ is
frozen then $\hat{u}_i = 0$. Otherwise the decoding rule is the following:
\begin{align*}
\hat{u}_{\pi(i)} = 
\left\{
\begin{array}{cc}
0, & \text{if $\frac{\Pr(y_{0}^{N-1} \mid \hat{u}_{\pi(0)}^{\pi(i-1)}, U_{\pi(i)} =
0)}{\Pr(y_{0}^{N-1}|\hat{u}_{\pi(0)}^{\pi(i-1)}, U_{\pi(i)}= 1)} > 1$},\\
1, & \text{otherwise}.
\end{array}\right.
\end{align*}

Using the factor graph representation between $\bar{u}$ and $\bar{x}$
shown in Figure~\ref{fig:differenttrellis}(a), Ar\i kan showed that this decoder
can be implemented with $O(N\log N)$ complexity.

\begin{theorem}\cite{ArT08}\label{thm:ArT}
Let $W$ be any symmetric B-DMC with $I(W) > 0$. Let $R < I(W)$ and $\beta < \frac12$ be
fixed. Then for $N=2^n$, $n\geq 0$, the probability of error for polar
coding under SC decoding at block length $N$ and rate $R$
satisfies
$P_e(N,R) = o(2^{-N^\beta}).$
\end{theorem}

\subsection{Optimality of the Exponent}
The following lemma characterizes the minimum distance of a polar
code.

\blemma[Minimum Distance of Polar Codes]\label{lem:dminpol}
Let $I$ be the set of information bits of a polar code $\code$. The
minimum distance of the code is given by 
$d_{\min}(\code) = \min_{i\in I} 2^{{\wt}(i)}.$
\elemma
\bproof
Let $w_{\min} = \min_{i\in I}\wt(i)$. Clearly, $d_{\min}$ cannot be
larger than the minimum weight of the rows of the generator matrix. Therefore,
$d_{\min} \leq 2^{w_{\min}}$.
On the other hand, by adding some extra rows to the generator matrix we cannot increase
 the minimum distance. In particular, add all the rows of $G_2^{\otimes n}$ with weight 
at least $2^{w_{\min}}$. The resulting code is $\RM(n,n-w_{\min})$. It is well
known that $d_{\min}(\RM(n,r)) = 2^{n-r}$ \cite{For01}.
Therefore,
$d_{\min} \geq d_{\min}(\RM(n,n-w_{\min})) = 2^{w_{\min}}.$
\eproof
We conclude that for any given rate $R$, if the information bits are 
picked according to
their weight (RM rule), i.e., the $2^nR$ vectors of largest weight, the resulting
code has the largest possible minimum distance. 
The following lemma gives a bound on the best possible minimum distance for any 
non-trivial rate.

\blemma\label{lem:mindistbnd}
For any rate $R > 0$ and any choice of information bits, the minimum distance 
of a code of length $2^n$ is bounded as 
$d_{\text{min}} \leq 2^{\frac{n}{2} + c \sqrt n}$
for $n > n_o(R)$ and  a constant $c=c(R)$.
\elemma
\begin{proof}
Lemma \ref{lem:dminpol} implies that $d_{\min}$ is maximized by choosing the 
frozen bits according to the RM rule.
The matrix $G_2^{\otimes n}$ has $n \choose i$ rows of weight $2^i$. Therefore, 
$d_{\min} \leq 2^{k} : \sum_{i=k+1}^{n}{{n}\choose{i}} < 2^nR \leq \sum_{i=k}^n{{n}\choose{i}}.$
For $R >  \frac12$, more than half of the rows are in the generator matrix. Therefore, there is
at least one row with weight less than or equal to $2^{\lceil{\frac{n}{2}}\rceil}$.
Consider therefore an $R$ in the range $(0, 1/2]$. 
Using Stirling's approximation, one can show that  
%
for any $R>0$
at least one row of the generator matrix has weight in the range $[2^{\lceil\frac{n}{2}\rceil-c\sqrt n},
2^{\lceil \frac{n}{2}\rceil+c\sqrt n}]$. Using Lemma~\ref{lem:dminpol} we conclude the
result.
\end{proof}

\begin{theorem}
Let $R>0$ and $\beta > \frac12$ be fixed.  
For any symmetric B-DMC $W$, and $N=2^n$, $n\geq n(\beta, R, W)$,
the probability of error for polar coding 
under MAP decoding at block length $N$
and rate $R$ satisfies 
$P_e(N,R) > 2^{-{N}^\beta}$
\end{theorem}
\begin{proof}
For a code with minimum distance $d_{\min}$, the block error probability is
lower bounded by $2^{-K d_{\min}}$ for some positive
constant $K$, which only depends on the channel.  This is easily seen by
considering a genie decoder; the genie provides the correct value of all
bits except those which differ between the actually transmitted codeword
and its minimum distance cousin.  Lemma~\ref{lem:mindistbnd} implies that
for any $R>0$, for $n$ large enough, $d_{\min} < \frac{1}{K} N^{\beta}$ for any $\beta > \frac12$.  Therefore
$P_e(N,R) > 2^{-{N}^\beta}$ for any $\beta > \frac12$.  
\end{proof}

This combined with Theorem~\ref{thm:ArT}, implies that the SC decoder
achieves performance comparable to the MAP decoder  in terms of the
order of the exponent.

\section{Performance under Belief Propagation}
Theorem~\ref{thm:ArT} does not state what lengths are
needed in order to achieve the promised rapid decay in the error probability nor does it specify the involved
constants. Indeed, for moderate lengths polar codes under
SC decoding are not record breaking. 
In this section we show various ways to improve the performance
of polar codes by considering belief propagation (BP) decoding. 
BP was already used in \cite{Ari08b} to compare the performance of  
polar codes based on Ar\i kan's rule and RM rule. 
For all the simulation points
in the plots the $95\%$ confidence intervals are shown. In most
cases these confidence intervals are smaller than the point size and are therefore
not visible.

\subsection{Successive Decoding as a Particular Instance of BP}
For communication over a binary erasure channel (BEC) one can easily show the following.
\begin{lemma}
Decoding the bit $U_i$ with the SC decoder is equivalent to applying BP with the knowledge of 
$U_0,\dots,U_{i-1}$ and all other bits unknown (and a uniform prior on them).
\end{lemma}

We conclude that if we use a standard BP algorithm (such a decoder
has access also to the information provided by the frozen bits
belonging to $U_{i+1},\dots,U_{N-1}$) then it is in
general strictly better than a SC decoder.
Indeed, it is not hard to construct explicit examples of codewords
and erasure patterns to see that this inclusion is strict (BP decoder
succeeds but the SC decoder does not). Figure~\ref{fig:BEC_BPcyclic}
shows the simulation results for the SC, BP and the MAP decoders when transmission 
takes place over the BEC. As we can see from these simulation results,
the performance of the BP decoder lies roughly half way between that of the
SC decoder and that of the MAP decoder.
For the BEC the {\em scheduling} of the individual messages is irrelevant
to the performance as long as each edge is updated until a fixed point has
been reached. For general B-DMCs the performance relies heavily on the
specific schedule. We found empirically that a good performance can be
achieved by the following schedule. Update the messages of each of the $n$
sections of the trellis from right to left and then from left to right and
so on.  Each section consists of a collection of $Z$ shaped sub-graphs.
We first update the lower horizontal edge, then the diagonal edge, and,
finally, the upper horizontal edge of each of these $Z$ sections.  In this
schedule the information is spread from the variables belonging to one
level to its neighboring level.  Figure~\ref{fig:AWGN} shows the simulation results
for the SC decoder and the BP decoder over the binary input additive
white Gaussian noise channel (BAWGNC) of capacity $\frac12$. Again,
we can see a marked improvement of the BP decoder over the SC decoder.
\begin{figure}[h!]
\centering
\input{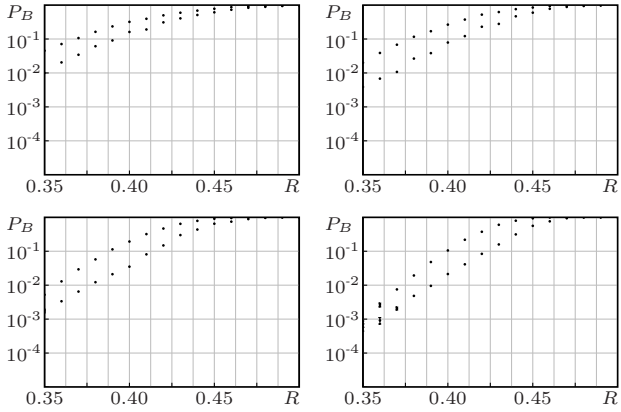}
\caption{\label{fig:AWGN}
Comparison of (i) SC and (ii) BP decoder in terms of block error 
probability, when transmission takes place over the BAWGNC$(\sigma=0.97865)$. 
The performance curves are shown for $n=10$ (top left), $11$ (top right), $12$ (bottom left), and $13$ (bottom right).
}
\end{figure}

\subsection{Overcomplete Representation: Redundant Trellises}
For the polar code of length $2^3$ one can check that all the three trellises
shown in Figure~\ref{fig:differenttrellis} are valid representations.
In fact, for a code of block length $2^n$, there
exist $n!$ different representations obtained by different permutations
of the $n$ layers of connections. Therefore, we can connect the vectors
$\bar{x}$ and $\bar{u}$ with any number of these representations and
this results in an {\em overcomplete} representation (similar to the concept
used when computing the stopping redundancy of a code \cite{ScV06}).
For the BEC any such overcomplete representation only improves the
performance of the BP decoder \cite{ScV06}. Further, the decoding complexity scales
linearly with the number of different representations used. Keeping
the complexity in mind, instead of considering all the $n!$ factorial
trellises, we use only the $n$ trellises obtained by cyclic shifts (e.g.,
see Figure~\ref{fig:differenttrellis}).  The complexity of this algorithm
is $O(N (\log N)^2)$  as compared to $O(N\log N)$ of the SC decoder and
BP over one trellis. The performance of the BP decoder is improved significantly by using this 
overcomplete representation as shown in Figure~\ref{fig:BEC_BPcyclic}.
\begin{figure}
\centering
\input{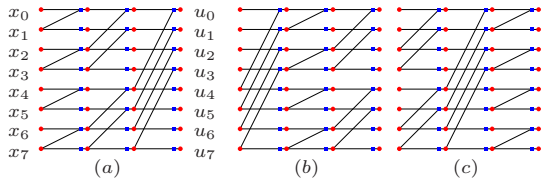}
\caption{\label{fig:differenttrellis} Figure (a) shows the trellis representation used by Ar\i
kan. Figures (b) and (c) show equivalent representations obtained by
cyclic shifts of the $3$ sections of trellis (a).}
\end{figure}
\begin{figure}[htp]
\centering
\input{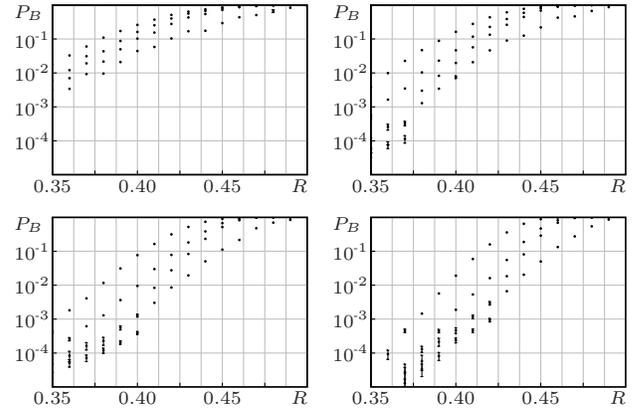}
\caption{\label{fig:BEC_BPcyclic}
Comparison of (i) SC, (ii) BP, (iii) BP with multiple trellises, and
(iv) MAP in terms of block error probability, when transmission takes place over
the BEC$(\epsilon=\frac12)$. The performance curves are shown for
$n=10$ (top left), $11$ (top right), $12$ (bottom left), $13$ (bottom right).
}
\end{figure}
We leave a systematic investigation of good schedules and choices of
overcomplete representations for general symmetric channels as an
interesting open problem.

\subsection{Choice of Frozen Bits}
For the  BP or MAP decoding algorithm the choice of frozen bits as given
by  Ar\i kan is not necessarily optimal.
In the case of MAP decoding we observe (see Figure~\ref{fig:BEC_MAP}) that 
the performance is significantly improved by picking the
frozen bits according to the RM rule. This is not a coincidence;
$d_{\text{min}}$ is maximized for this choice.
This suggests that there might be a  rule which is optimized for BP. 
It is an interesting open question to find such a rule.
\begin{figure}[htp]
\centering
\input{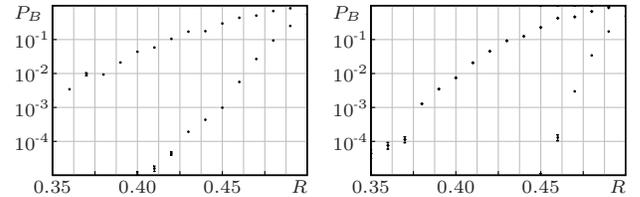}
\caption{\label{fig:BEC_MAP}
Comparison of block error probability curves under MAP decoding between codes  
picked according to Ar\i kan's rule and RM rule. The performance curves are
shown for $n=10$ (left) and $11$ (right).}
\end{figure}

\section{Source Coding}

In this section we show the performance of polar codes in the context of
source coding. We consider both lossless and lossy cases and show (in most cases) empirically
that they achieve the optimal performance in both 
cases. Let Ber$(p)$ denote a Bernoulli source with $\Pr(1) = p$.
Let $h_2(\cdot)$ denote the binary entropy function and $h_2^{-1}(\cdot)$
its inverse.

\subsection{Lossless Source Coding}
\subsubsection{Single User}
The problem of lossless source coding of a Ber$(p)$ source can be mapped to the
channel coding problem over a binary symmetric channel (BSC) as shown in 
\cite{AlB72,Wei62}. Let
$\bar{x} = (x_0,\dots,x_{N-1})^T$ be a sequence of $N$ i.i.d. realizations of
the source.
Consider a code of rate $R$ represented by the parity check matrix $\mathbf{H}$. The vector
$\bar{x}$ is encoded by its syndrome $\bar{s} = \mathbf{H}\bar{x}$. The rate of the resulting source code is
$1-R$. The decoding problem is to estimate 
$\bar{x}$ given the syndrome $\bar{s}$. This is equivalent to estimating a noise vector
in the context of channel coding over BSC$(p)$. Therefore, if a
sequence of codes achieve capacity over BSC($p$), then the corresponding
source codes approach a rate $h_2(p)$ with vanishing error probability.

We conclude that polar codes achieve the Shannon bound for
lossless compression of a binary memoryless source. Moreover, using 
the trellis of Figure~\ref{fig:differenttrellis}(a), we can compute the syndrome with complexity $O(N\log N)$. 
The source coding problem has a
considerable advantage compared to the channel coding problem. 
The encoder knows the information seen
by the decoder (unlike channel coding there is no noise involved
here).  Therefore, the encoder can also decode and check whether
the decoding is successful or not. In case of failure, the encoder
can retry the compression procedure by using a permutation of the source vector.
This permutation is fixed a priori and is known both to the encoder as well
as the decoder. In order to completely specify the system, the encoder
must inform the decoder which permutation was finally used.
This results in a small loss of rate but it brings down the probability of
decoding failure. Note that the extra number of bits that need to be transmitted grows
only logarithmically with the number of permutations used, but that the error probability
decays exponentially as long as the various permuted source vectors look like independent
source samples.  
With this trick one can make the curves essentially arbitrarily steep
with a very small loss in rate (see Figure~\ref{fig:SOURCE}). 

\begin{figure}[h!]
\centering
\input{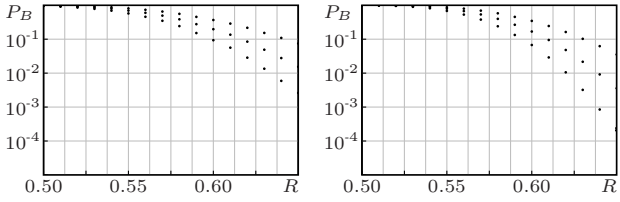}
\caption{\label{fig:SOURCE}
Comparison of SC decoding for a Ber$(0.11)$ source with 0,1,
and 2 bits for permutations. The performance curves are
shown for $n=10$ (left) and $11$ (right). By increasing the number
of permutations the curves can be made steeper and steeper. 
}
\end{figure}
\subsubsection{Slepian-Wolf}
Consider two Ber$(\frac12)$ sources $X$ and $Y$. Assume that
they are correlated as $X=Y\oplus Z$, where $Z\sim$ Ber$(p)$. 
%
Recall that the Slepian-Wolf rate region is the unbounded polytope
described by $R_X > H(X) =1$, $R_Y>H(Y) =1$, $R_X + R_Y > H(X,Y) =
1+h_2(p)$.  The points $(R_X,R_Y)=(1,h_2(p))$ and $(R_X,R_Y)=(h_2(p),1)$
are the so-called {\em corner points}.  
Because of symmetry
it suffices to show how to achieve one such corner point (say $(1,h_2(p))$). 

Let $\bar{x}$ and $\bar{y}$ denote $N$ i.i.d. realizations of the two
sources. The scheme using linear codes is the following: The encoder for
$X$ transmits $\bar{x}$ as it is (since $H(X)=1$, and so no compression
is necessary). Let $\mathbf{H}$ denote the parity-check matrix of a code designed
for communication over the BSC$(p)$. The encoder for $Y$ computes the
syndrome $\bar{s} = \mathbf{H}\bar{y}$ and transmits it to the receiver. At
the receiver we know $\bar{s} = \mathbf{H}\bar{y}$ and $\bar{x}$, therefore we can compute
$\bar{s}' = \mathbf{H}(\bar{x} \oplus \bar{y}) = \mathbf{H}\bar{z}$. The resulting problem of estimating $\bar{z}$ 
is equivalent to the lossless compression of a Ber$(p)$ discussed in the previous section.
Therefore, once again polar codes provide an efficient solution.
The error probability curves under
SC decoding are equal to the curves shown in the Figure~\ref{fig:SOURCE}
with $0$ bits for permutations.  

\subsection{Lossy Source Coding}
\subsubsection{Single User}
We do not know of a mapping that converts the lossy source coding
problem to a channel coding problem. However, for the binary erasure
source considered in \cite{MaY03}, it was shown how to construct a
``good'' source code from a ``good'' (over the BEC) channel code. We briefly
describe their construction here and show that polar codes
achieve the optimal rate for zero distortion.

The source is a sequence of i.i.d. realizations of a random variable $S$ taking
values in $\{0,1,\ast\}$ with $\Pr(S = 0) = \Pr(S=1) = \frac12 (1-\epsilon),
\Pr(S=\ast) = \epsilon$.
The reconstruction alphabet is $\{0,1\}$ and
the distortion function is given by 
$
\disto(\ast,0) = \disto(\ast,1) = 0, \disto(0,1) = 1.
$
For zero distortion, the rate of the rate-distortion function is given by
$R(D=0) = 1-\epsilon$. In \cite{MaY03} it was shown that the dual of a sequence of
channel codes which achieve the capacity of the BEC($1-\epsilon$) 
under BP decoding, achieve the rate-distortion pair for zero distortion 
using a message passing algorithm which they refer to as the erasure 
quantization algorithm. Polar codes achieve capacity under 
SC decoding. For communication over BEC, the performance under BP 
is at least as good as SC. Therefore, the dual of the polar codes 
designed for BEC$(1-\epsilon)$ achieve the optimum rate for zero 
distortion using the erasure quantization algorithm. Here, we show 
that the dual polar codes achieve the optimum rate for zero 
distortion even under a suitably defined SC decoding algorithm. 

The dual of a polar code is obtained by reversing the
roles of the check and variable nodes of the trellis in
Figure~\ref{fig:differenttrellis}(a) and reversing the roles
of the frozen and free bits. It is easy to see that $G_2^{\otimes n}
G_2^{\otimes n} = \mathbf{I}$. This implies that the dual of a polar code is also a
polar code. The {\em suitably} defined algorithm is given by SC
decoding in the order $\pi({N-1}),\dots,\pi(0)$, opposite of the original decoding
order. We refer to this as the dual order.


In \cite{Ari08}, the probability of erasure for bit $u_i$ under SC decoding is given by
$Z_n^{(i)}(1-\epsilon)$, computed as follows. 
Let the $n$-bit binary  expansion of
$i$ be given by $b_{n-1},\dots,b_0$. Let $Z_0 = 1-\epsilon$. The sequence
$Z_1,\dots,Z_n=Z_n^{(i)}(1-\epsilon)$ is recursively defined as follows:
\begin{align}\label{eqn:Z}
Z_k = \left\{ \begin{array}{lc}
Z_{k-1}^2, & \text{ if } b_{k-1} = 1,\\
1-(1-Z_{k-1})^2, & \text{ if } b_{k-1} = 0.
\end{array}\right.
\end{align}

From
\cite{ArT08} we know that for any $R<\epsilon$, there exists an $n_0$ such
that for $n \geq n_0$ we can find a set $I\subset \{0,\dots,2^{n}-1\}$ of
size $2^nR$ satisfying $\sum_{i\in I}Z_n^{(i)}(1-\epsilon) \leq 2^{-N^\beta}$ for any $\beta
< \frac12$. The set $I$ is used as information bits. The complement of $I$
denoted by $F$, is the set of frozen bits.

Let $\Zd_n^{(i)}(\epsilon)$ denote the probability of erasure for bit $u_i$ of the dual code
used for BEC$(\epsilon)$. One can check that for the dual code with the dual
order, the bit $u_i$ is equivalent to the bit $u_{N-i}$ of the original code
with the original decoding order. Let $\Zd_0 = \epsilon$. The recursive computation for
$\Zd_n^{(i)}$ is given by 
\begin{align}\label{eqn:dualZ}
\Zd_k = \left\{ \begin{array}{lc}
\Zd_{k-1}^2, & \text{ if } b_{k-1} = 0,\\
1-(1-\Zd_{k-1})^2, & \text{ if } b_{k-1} = 1.
\end{array}\right.
\end{align}
\begin{lemma}[Duality for BEC]
$Z_n^{(i)}(1-\epsilon) + \Zd_n^{(i)}(\epsilon) = 1$. 
\end{lemma}
The proof follows through induction on $n$ and using the equations $\eqref{eqn:Z}$ and $\eqref{eqn:dualZ}$.
For the dual code the set $I$ (information set for BEC$(1-\epsilon)$) is used as frozen bits and the 
set $F$ is used as information bits. Let $\bar{x}$ be a 
sequence $2^n$ source realizations. The source vector $\bar{x}$ needs to be mapped to 
a vector $\bar{u}_{F}$ 
(information bits) such that $\disto(\bar{u}G_2^{\otimes n},\bar{x}) = 0$. 
The following lemma shows that such a vector $\bar{u}_F$ can be found using SC decoder 
with vanishing error probability.

\begin{lemma}
The probability of encoding failure for erasure quantization of the source $S$ using the dual 
of the polar code designed for the BEC$(1-\epsilon)$ and SC decoding with dual order
is bounded as $P_e(N) = o(2^{-N^\beta})$ for any $\beta < \frac{1}{2}$.
\end{lemma}
\begin{proof}
Let $I$ and $F$ be as defined above.
The bits belonging to $I$ are already fixed to $0$ whereas the bits belonging to $F$ are free
to be set.  Therefore an error can only occur if one of the bits belonging to $I$ are set
to a wrong value. However, if the SC decoding results in an erasure for these bits,
these bits are also free to be set any value and this results in no error.
Therefore the probability of error can be upper bounded by the probability that at least one of
the bits in $I$ is not an erasure which in turn can be upper bounded by $\sum_{i\in
I}(1-\Zd_n^{(i)}(\epsilon)) = \sum_{i \in I}Z_n^{(i)}(1-\epsilon)
= o(2^{-N^\beta})$, where the last equality follows from Theorem~\ref{thm:ArT}.
\end{proof}

The fact that the dual code can be successfully applied to the
erasure source, suggests to extend this construction to more general
sources.  Let us try a similar construction to encode the Ber($\frac12$)
source with the Hamming distortion. To design a source code for
distortion $D$ we first design a rate $1-h_2(p)$ polar code
for the BSC$(p)$ where $p=h_2^{-1}(1-h_2(D))$. The design consists of choosing the generator
vectors from the rows of $G_2^{\otimes n}$ as explained in \cite{Ari08}.
The source code is then defined by the corresponding dual code with
the dual decoding order.
Since the rate of the original
code is $1-h_2(p)$, the rate of the dual code is $h_2(p) = 1-h_2(D)$.
Figure~\ref{fig:LOSSY_WYNER} shows the rate-distortion performance of
these codes for various lengths. As the lengths increase, empirically
we observe that the performance approaches the rate-distortion
curve.
\begin{figure}[htp]
\centering
\input{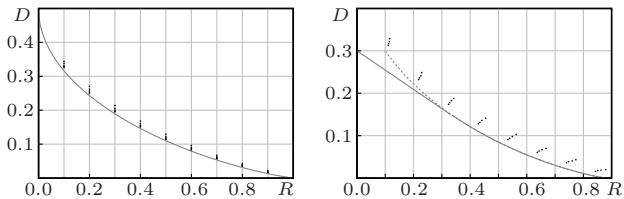}
\caption{\label{fig:LOSSY_WYNER}
(left) Rate-Distortion curves for a Ber$(0.11)$ source. The performance curves
are shown for $n = 9,11,13,15,17$.(right) The dashed line is the curve $R_{WZ}(D)$. 
The solid line is the lower convex envelope of  $R_{WZ}(D)$ and 
$(0,0.3)$ for a Ber$(\frac12)$ source with receiver having a side information obtained
through a BSC$(0.3)$. 
The performance curves are shown for $n =
9, 11, 13, 15$.}
\end{figure}
It is an interesting and challenging open problem to prove this observation rigorously.

\subsubsection{Wyner-Ziv}
Let $X$ denote a Ber$(\frac12)$ source which we want to compress.
The source is reconstructed at a receiver which has access to a
side information $Y$ correlated  to $X$ as $Y=X\oplus Z$ with $Z\sim$
Ber$(p)$. Let $\bar{x}$ and $\bar{y}$ denote a sequence of $N$ i.i.d. realizations 
of $X$ and $Y$. Wyner and Ziv have shown that the rate
distortion curve is given by the lower convex envelope of the curve
$R_{WZ}(D) = h_2(D\ast p) - h_2(D)$ and the point $(0,p)$, where $D \ast p
= D(1-p) + (1-D)p$.  As discussed in \cite{ZSE02}, nested linear
codes are required to tackle this problem.  The idea is to partition
the codewords of a code $\code_1$ into cosets of another code
$\code_2$. The code $\code_1$ must be a good source code and the
code $\code_2$ must be a good channel code.

Using polar codes, we can create this nested structure as follows.
Let $\code_1$ be a source code for distortion $D$ (Bernoulli source,
Hamming distortion) as described in the previous section.  Let  $F_s$
be the set of frozen (fixed to $0$) bits for this source code. Using this
code we quantize the source vector $\bar{x}$. Let the resulting vector be
$\hat{x}$. Note that the vector $\hat{x}$ is given by $\hat{u}G_2^{\otimes
n}$, where $\hat{u} = (\hat{u}_0,\dots,\hat{u}_{N-1})$ is such that
$\hat{u}_i = 0$ for $i \in F_s$ and for $i\in F_s^c$, $\hat{u}_i$ is
defined by the source quantization.

For a sequence of codes achieving the rate-distortion bound, it was shown in
\cite{ErZ02} that the ``noise'' $\bar{x} \oplus \hat{x}$ added due to the
quantization is comparable to a Bernoulli noise Ber$(D)$. Therefore, for
all practical purposes, the side information at the receiver $\bar{y}$
is statistically equivalent to the output of $\hat{x}$ transmitted
through a BSC$(D\ast p)$. 
Let $F_c$ be the set of frozen bits of a channel code for the BSC$(D\ast
p)$.  Let the encoder transmit the bits $\hat{u}_{F_c\backslash
F_s}$ to the receiver. Since the bits belonging to $F_s$ ($\hat{u}_{F_s}$)
are fixed to zero, the receiver now has access to $\hat{u}_{F_c}$
and $\bar{y}$.  The definition of $F_c$ implies that the receiver
is now able to decode $\hat{u}_{F_c^c}$ (and hence $\hat{x}$) with
vanishing error probability. To see the nested structure, note that
the code $\code_2$ and its cosets are the different channel codes
defined by different values of $\bar{u}_{F_c\backslash F_s}$ and that
these codes partition the source code $\code_1$.
The capacity achieving property of polar codes implies
that for $N$ sufficiently large, $|F_c| < N(h_2(D\ast p) + \delta)$ for
any $\delta > 0$. In addition, if we assume that polar codes
achieve the rate-distortion bound as conjectured in the previous section
and $F_s \subseteq F_c $, then for $N$ sufficiently large, the rate
required to achieve a distortion $D$ is $\frac{1}{N}|F_c \backslash F_s|
\leq h_2(D\ast p) - h_2(D) + \delta$ for any $\delta > 0$.  This would show
that polar codes can be used to efficiently realize the Wyner-Ziv
scheme. The performance of the above scheme for various lengths is shown
in Figure~\ref{fig:LOSSY_WYNER}.


\bibliographystyle{IEEEtran}
\bibliography{lth,lthpub}

\newcommand{\SortNoop}[1]{}
\begin{thebibliography}{10}
\providecommand{\url}[1]{#1}
\csname url@rmstyle\endcsname
\providecommand{\newblock}{\relax}
\providecommand{\bibinfo}[2]{#2}
\providecommand\BIBentrySTDinterwordspacing{\spaceskip=0pt\relax}
\providecommand\BIBentryALTinterwordstretchfactor{4}
\providecommand\BIBentryALTinterwordspacing{\spaceskip=\fontdimen2\font plus
\BIBentryALTinterwordstretchfactor\fontdimen3\font minus
  \fontdimen4\font\relax}
\providecommand\BIBforeignlanguage[2]{{%
\expandafter\ifx\csname l@#1\endcsname\relax
\typeout{** WARNING: IEEEtran.bst: No hyphenation pattern has been}%
\typeout{** loaded for the language `#1'. Using the pattern for}%
\typeout{** the default language instead.}%
\else
\language=\csname l@#1\endcsname
\fi
#2}}

\bibitem{Ari08}
E.~{Ar\i kan}, ``Channel polarization: A method for constructing
  capacity-achieving codes for symmetric binary-input memoryless channels,''
  \emph{submitted to IEEE Trans. Inform. Theory}, 2008.

\bibitem{For01}
G.~D. Forney, Jr., ``Codes on graphs: Normal realizations,'' \emph{IEEE Trans.
  Inform. Theory}, vol.~47, no.~2, pp. 520--548, Feb. 2001.

\bibitem{ArT08}
E.~{Ar\i kan} and E.~{Telatar}, ``{On the rate of channel polarization},'' July
  2008, available from ``http://arxiv.org/pdf/0807.3917''.

\bibitem{Ari08b}
E.~{Ar\i kan}, ``A performance comparison of {P}olar codes and {R}eed-{M}uller
  codes,'' \emph{IEEE Communications Letters}, vol.~12, no.~6, 2008.

\bibitem{ScV06}
M.~Schwartz and A.~Vardy, ``On the stopping distance and the stopping
  redundancy of codes,'' \emph{IEEE Trans. Inform. Theory}, vol.~52, no.~3, pp.
  922 -- 932, Mar. 2006.

\bibitem{AlB72}
P.~E. Allard and E.~W. Bridgewater, ``A source encoding technique using
  algebraic codes,'' June 1972, pp. 201--203.

\bibitem{Wei62}
E.~Weiss, ``Compression and coding,'' \emph{IRE Trans. Inform. Theory}, pp.
  256--257, 1962.

\bibitem{MaY03}
E.~Martinian and J.~Yedidia, ``Iterative quantization using codes on graphs,''
  in \emph{Proc. of the Allerton Conf. on Commun., Control, and Computing},
  Monticello, IL, USA, 2003.

\bibitem{ZSE02}
R.~Zamir, S.~Shamai, and U.~Erez, ``Nested linear/lattice codes for structured
  multiterminal binning,'' \emph{IEEE Transactions on Information Theory},
  vol.~48, no.~6, pp. 1250--1216, 2002.

\bibitem{ErZ02}
U.~Erez and R.~Zamir, ``Bounds on the $\epsilon$-covering radius of linear
  codes with applications to self noise in nested {W}yner-{Z}iv coding,'' Feb.
  2002, dept. Elec. Eng-Syst., Tel-Aviv Univ., Tech. Rep.

\end{thebibliography}
\end{document}